\documentclass{llncs}
\usepackage{amsmath}
\usepackage{graphicx}
\usepackage [ruled,linesnumbered]{algorithm2e}

\newcommand{\C }{\mathbf{C}}
\newcommand{\R }{\mathbf{R}}

\begin{document}

\title{Finding minimum Tucker submatrices}
\author{J\'an Ma\v nuch \inst{1,2} \and Arash Rafiey\inst{2}}

\institute{Department of Computer Science, UBC, Vancouver, BC, Canada
  \and Department of Mathematics, Simon Fraser University,
  Burnaby, BC, Canada\\
\email{jmanuch@cs.ubc.ca,arashr@sfu.ca}
}

\maketitle

\begin{abstract}
A binary matrix has the Consecutive Ones Property (C1P) if its columns can be ordered in such a way that all 1s on each row are consecutive.
These matrices are used for DNA physical mapping and ancestral genome reconstruction in computational biology on the
other hand they represents a class of convex bipartite graphs and are of interest of algorithm graph theory researchers.
Tucker gave a forbidden submartices characterization of matrices that have C1P property in 1972. Booth and Lucker (1976) gave a
first linear time recognition algorithm for matrices with C1P property and then in 2002, Habib, et al. gave a simpler linear time recognition algorithm.
There has been substantial amount of works on efficiently finding minimum size forbidden submatrix.
Our algorithm is at least $n$ times faster than the existing algorithm where $n$ is the number of columns of the input matrix.

\end{abstract}

\section{Introduction and Preliminaries}
\label{sec:preliminaries}

A binary matrix has the Consecutive Ones Property (C1P) if its columns
can be ordered in such a way that all ones in each row are
consecutive. Deciding if a matrix has the C1P can be done in
linear-time and
space~\cite{BoothLueker1976,habib-lex,hsu-simple,mcconnell-certifying,meidanis-on}.
The problem of deciding if a matrix has the C1P has been considered in
genomic, for problems such as physical
mapping~\cite{alizadeh-physical,lu-test} or ancestral
genome reconstruction~\cite{adam-modelfree,chauve-methodological,ma-reconstructing}.

Let $M$ be a $m\times n$ binary matrix. Let $\R =
\{r_{i}:\ i = 1,\dots,m\}$ be the set of its rows and $\C = \{c_{j}:\
j = 1,\dots,n\}$ the set of its columns. Its \emph{corresponding
  bipartite graph} $G(M) = (V_{M},E_{M})$ is defined as follows:
$V_{M} = \R\cup \C$, and two vertices $r_{i}\in R$ and $c_{j}\in \C$
are connected by an edge if and only if $M[i,j] = 1$. We will refer to
the partition $\R$ and $\C$ of $G(M)$ as black and white vertices,
respectively. The set of neighbors of a vertex $x$ will be denoted by
$N(x)$. The $i$-the neighborhood of $x$, denoted by $N_{i}(x)$, is the
set of vertices distance $i$ from $x$. All these sets, for a fixed
$x$, can be computed in time $O(e)$ using the bread-first search
algorithm. A subgraph of $G(M)$ induces by vertices
$x_{1},\dots,x_{k}$ will be denoted by $G(M)[x_{1},\dots,x_{k}]$. A
set of edges of bipartite graph is called \emph{induced matching} if
the set of endpoints of these edges induces this matching in the
graph. For example, two edges $\{u,v\}$ and $\{u',v'\}$, where $u,u'$
are in the same partition form an induced matching if $\{u,v'\}$ and
$\{u',v\}$ are not edges of the graph.

An \emph{asteroidal triple} is an independent set of three vertices such that
each pair is connected by a path that avoids the neighborhood of the
third vertex. A \emph{white asteroidal triple} is an asteroidal triple
on white (column) vertices.

The following result of Tucker links the C1P of matrices to asteroidal
triples of their bipartite graphs.

\begin{theorem}[\cite{Tucker1972}]
  A binary matrix has the C1P if and only if its corresponding
  bipartite graph does not contain any white asteroidal
  triples.
\end{theorem}

\begin{theorem}[\cite{Tucker1972}]
  A binary matrix has the C1P
  if and only
  if its corresponding bipartite graph does not contain any of the
  forbidden subgraphs in $T =
  \{G_{\mathrm{I}_{k}},G_{\mathrm{II}_{k}},G_{\mathrm{III}_{k}}:\ k\ge
  1\} \cup \{G_{\mathrm{IV}},G_{\mathrm{V}}\}$, depicted in
  Figure~\ref{fig:forbidden-subgraphs}. We will refer to these
  subgraphs as the type I, II, III, IV and V, respectively.
\end{theorem}

\begin{figure}
  \centering
  \includegraphics[width = .7\linewidth ]{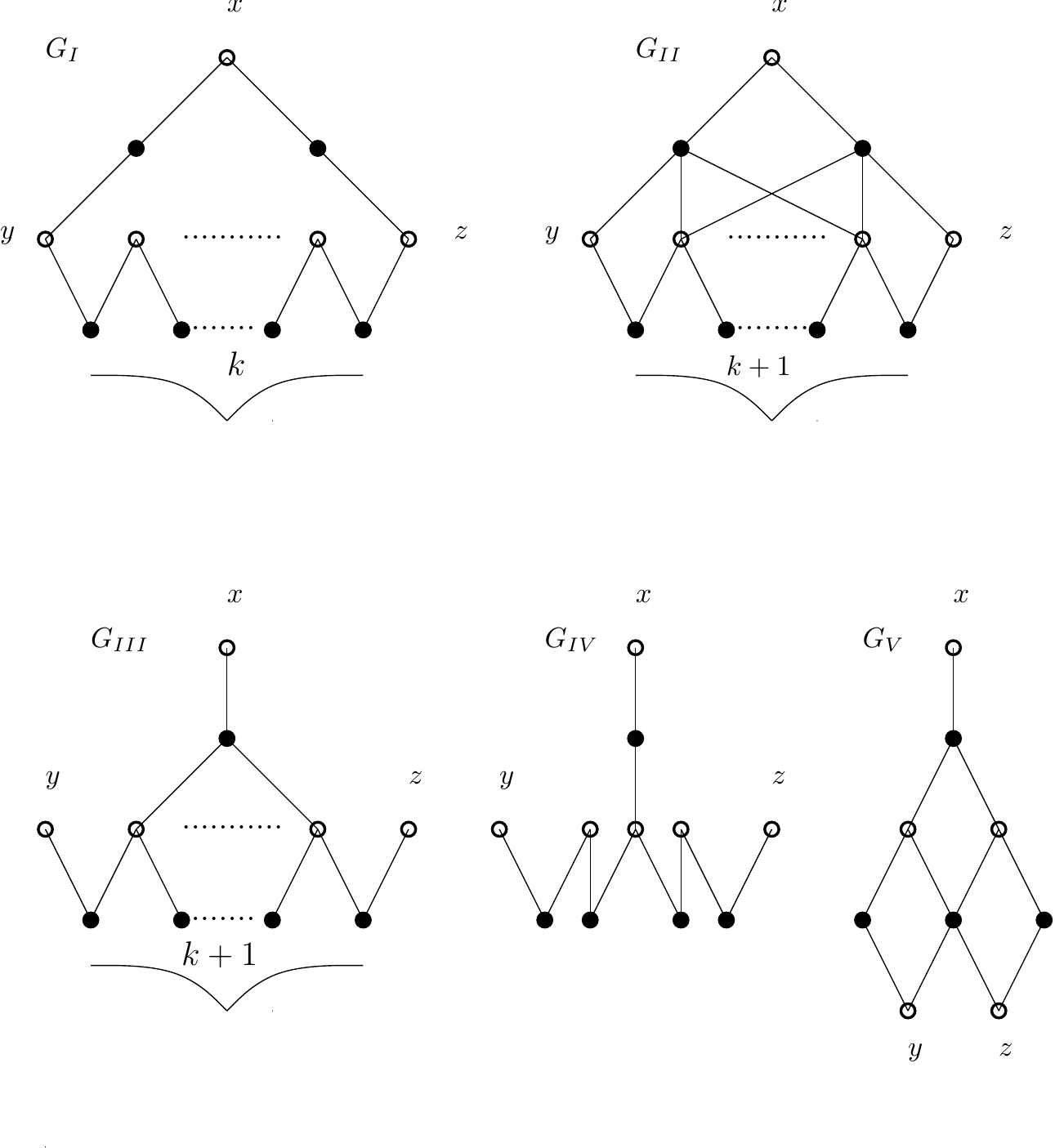}
  \caption{The set of Tucker's forbidden subgraphs.}
  \label{fig:forbidden-subgraphs}
\end{figure}

The author in \cite{DBLP:conf/wg/LindzeyM13} developed an algorithm for finding one of the obstructions
in linear time. However, their algorithm does not guarantee the minimum size obstruction.
The characterization can be used to determine whether a given binary
matrix has the C1P in time $O(\Delta mn^{2} + n^{3})$, where $\Delta $
is the maximum number of ones per row, i.e., the maximum degree of
black vertices in $G(M)$, as explained by the following result in
\cite{DomGuoNiedermeier2010}.

\begin{lemma}[\cite{DomGuoNiedermeier2010}]
  \label{l:asteroidal}
  A white asteroidal triple $u,v,w$ with the smallest sum of the three
  paths (avoiding the third neighborhood) can be computed in time
  $O(\Delta mn^{2} + n^{3})$.
\end{lemma}

For practical purposes, there is a much faster algorithm that uses
PQ-trees for determining whether a binary matrix has the C1P,
cf. \cite{BoothLueker1976}. Tucker's interest was in finding the
smallest submatrix of a non-C1P binary matrix which makes this matrix
non-C1P. He further refined his asteroidal triple characterization
using a set of \emph{forbidden submatrices}. We will state this
results in terms of \emph{forbidden subgraphs}.


We will consider two problems: (1) detected a smallest forbidden
subgraph of each type (Section~\ref{sec:detect-small-forb-each}), and
(2) detecting a smallest forbidden subgraph of any type
(Section~\ref{sec:detect-small-forb-all}).

We use the followings to improve the complexity :
\begin{itemize}
\item In our computation we use degree of each vertex instead the
  maximum degree $\Delta$.
\item
  We compute some of the necessary sets in advance.
\item
  In our analysis we use the minimum obstruction assumption and
  explore the connection of vertices around a minimum obstruction with
  it.
\end{itemize}

\setlength{\tabcolsep}{5pt}
\begin{table}
  \centering
  \begin{tabular}{c | l l l}
    Subgraph type & \multicolumn{3}{c}{Time complexity} \\
    & Previous result & Our result (Exact) & Our result\\
    \hline
    I & $O(\Delta^{4}m^{3})$ \cite{BlinRizziVialette2012} &
    $O(\Delta e^{2}) = O(\Delta^{3}m^{2})$ &
    $O(n^{2}e)$ \cite{DomGuoNiedermeier2010}\\
    II & $O(\Delta^{4}m^{3}) = O(ne^{3})$ \cite{BlinRizziVialette2012} &
    --- &
    $O(n^{2}e)$ \cite{DomGuoNiedermeier2010} \\
    III & $O(\Delta^{2}m^{2}n^{2})$ \cite{BlinRizziVialette2012} &
    $O(e^{3}) = O (\Delta^{3}m^{3})$ &
    $O(ne^{2})$ \\
    IV & $O(\Delta^{3}m^{2}n^{3})$ \cite{DomGuoNiedermeier2010} &
    $O(m^{3}e) = O(\Delta m^{4})$ &
    $O(n^{3}e)$ \\
    V & $O(\Delta^{4}m^{2}n )$ \cite{DomGuoNiedermeier2010} &
    $O(m^{3}e) = O(\Delta m^{4})$ &
    $O(n^{3}e)$ \\[3mm]
    Any & $O(\Delta^{3}m^{2}(\Delta m + n^{3}))$ &
    \multicolumn{2}{l}{
    $O(ne(n^{2} + e)) = O(\Delta mn(\Delta m + n^{2}))$}
  \end{tabular}
  \caption{Comparison of our results with the previous results.}
  \label{tab:comparison}
\end{table}

Note that without loss of generality we can assume that $M$ does not
contain any all-zero columns or rows, as such columns does not affect whether
the matrix has the C1P or the forbidden submatrices of $M$. It follows
that $\Delta m\ge n$. We will use this assumption throughout this
paper. Also note that the number of edges in $G(M)$ is the same as the
number of ones in $M$, which we denote as $e$. Note that $e = O(\Delta
m)$ and that $e\ge m,n$ (since we assume that there are no all-zero
columns or rows in $M$).

We will use the following auxiliary lemma.

\begin{lemma}\label{l:induced-matching-2}
  Given a bipartite graph $G$ with $e$ edges and partitions of size
  $m$ and $n$, picking an induced matching of size two of $G$ or
  determining that no such induced matching exists can be done in time
  $O(e + m + n)$.
\end{lemma}

\begin{proof}
  Let $U$ be the partition of size $n$. Order vertices of $U$ by their
  degrees:
  $\deg (u_{1})\le \deg (u_{2})\le\dots \le \deg (u_{n})$. For every
  $i = 1,\dots,n - 1$, check if $N(u_{i})\setminus N(u_{i + 1})$ is
  non-empty. If for some $i$, $N(u_{i})\setminus N(u_{i + 1})\ne
  \emptyset $, then also $N(u_{i + 1})\setminus N(u_{i})\ne \emptyset$.
  In this case, we can pick any $a\in N(u_{i})\setminus N(u_{i + 1})$
  and any $b\in N(u_{i + 1})\setminus N(u_{i})$, and return
  $\{u_{i},a\}$ and $\{u_{i + 1},b\}$, as it forms an induced
  matching of $G$.

  Now, assume that for every $i$, $N(u_{i})\setminus N(u_{i + 1}) =
  \emptyset $, i.e., $N(u_{i})\subseteq N(u_{i + 1})$. We will show
  that there is no induced matching of $G$ of size two. Assume for
  contradiction that $\{u_{i},a\}$ and $\{u_{j},b\}$, where $i < j$,
  is such an induced matching. We have $N(u_{i})\subseteq N(u_{i +
    1})\subseteq\dots \subseteq N(u_{j})$, i.e., $a\in N(u_{j})$, a
  contradiction. Hence, in this case we can report that there is no
  such matching.

  Vertices of $U$ can be sorted by their degrees in time $O(n + m)$
  using a count sort. For each $i$, checking if $N(u_{i})\setminus N(u_{i +
    1})$ is non-empty can be done in time $O(\deg (u_{1}))$, hence,
  the total time spent on checking is $O(\sum_{i = 1}^{n - 1} \deg
  (u_{i})) = O(e)$.
\end{proof}

\section{Detection of smallest forbidden subgraphs for each type}
\label{sec:detect-small-forb-each}

We will present four algorithms which find a smallest subgraph of type
I, III, IV and V, respectively, each improving the complexity of the
best known such algorithm, cf.~\cite{BlinRizziVialette2012}. For type
II, we refer reader to the $O(ne^{3})$ algorithm\footnote{The authors of
  \cite{BlinRizziVialette2012} showed that the complexity of their
  algorithm is $O(\Delta^{4}m^{3})$, however, it is easy to check that
  their algorithm works in time $O(ne^{3})$.}
in~\cite{BlinRizziVialette2012}.

\subsection{Type I}
\label{sec:type-i}

Algorithm~\ref{alg:type1} finds a smallest forbidden subgraph of type I
in time $O(\Delta e^{2})$.

\begin{algorithm}
  \caption{Find a smallest $G_{\mathrm{I}_{k}}$ subgraph.}
  \label{alg:type1}
  \scriptsize
  \SetKwInOut{Input}{Input}
  \SetKwInOut{Output}{Output}
    \Input {$G(M)$}
    \Output {A smallest subgraph $G_{\mathrm{I}_{k}}$ of $G(M)$}
    \BlankLine
    \For{$w\in \R$}{
      \For{$x,y\in N(w)$}{
        construct the subgraph $G_{w,x,y}$ of $G(M)$ induced by
        vertices $(\R\setminus (N(x)\cap N(y)))\cup (\C\setminus
        N(w))\cup \{x,y\}$\;
        find a shortest path between $x$ and $y$ in $G_{w,x,y}$\;
        \If {the length of the path is smaller than any observed so
          far}{
          remember $w$ and the vertices of the path\;
        }
      }
    }
    \Return {subgraph of $G(M)$ induced by the remembered set of
      vertices  (if any)}
\end{algorithm}

\emph{Correctness of Algorithm~\ref{alg:type1}.} We are looking for
induced cycles of length 6 or more. For each black vertex $w$ and its
two neighbors $x,y$, we find a shortest induced cycle of length at
least 6. Such cycle cannot contain any vertex incident with $w$ other
than $x$ and $y$, and any vertex incident with both $x$ and $y$ other
than $w$. Hence, a shortest such cycle $c$ can be obtained from the a
shortest $x - y$ path $p$ in $G_{w,x,y}$ by adding two edges $\{x,w\}$
and $\{y,w\}$. This cycle cannot be of length $4$, otherwise $p$ would
contain a vertex in $N(x)\cap N(y)$. It remains to show that $c$ is
induced. Assume that there is a chord $\{u,v\}$ in $c$. Since $p$ does
not contain $N(w)\setminus \{x,y\}$, $u,v\ne w$. Hence, we could use
the chord as a shortcut to find a shorter cycle containing edges
$\{x,w\}$ and $\{y,w\}$, and hence, a shorter path between $x$ and $y$
in $G_{w,x,y}$, a contradiction.

\emph{Complexity of Algorithm~\ref{alg:type1}.} We will show that the
complexity of Algorithm~\ref{alg:type1} is $O(\Delta e^{2}) =
O(\Delta^{3}m^{2})$. The first loop executes $m$ times and the second
$\deg (w)^{2}$ times. Hence, the body of the second loop executes
$\sum_{w\in \R } \deg (w)^{2} = O(\Delta e)$ times. Constructing graph
$G_{w,x,y}$ takes time $O(e)$ and finding a shortest path in
$G_{w,x,y}$ can be done in time $O(e)$ using the Breadth-first search
algorithm.

\subsection{Type III}
\label{sec:type-iii}



Algorithm~\ref{alg:type3a} finds a smallest forbidden subgraph of type
II in time $O(e^{3})$.

\begin{algorithm}
  \caption{Find a smallest $G_{\mathrm{III}_{k}}$ subgraph.}
  \label{alg:type3a}
  \SetKwInOut{Input}{Input}
  \SetKwInOut{Output}{Output}
  \scriptsize
  \Input {$G(M)$}
  \Output {A smallest subgraph $G_{\mathrm{III}_{k}}$ of $G(M)$}
  \BlankLine
  \For{$\{x,w\}\in E_{M}$, where $x\in \C$ and $w\in \R$}{
    \For{$\{y,a\}$, where $y\in \C\setminus N(w)$ and $a\in \R \setminus N(x)$}{
      construct the subgraph $G_{x,w,y,a}$ of $G(M)$ induced by vertices
      $(\R\setminus N(x)\setminus N(y))\cup \{a\} \cup (N(w)\setminus
      \{x\})\cup (\C \setminus N(a))$\;
      find a shortest path between $a$ and set $\C\setminus
      N(w)\setminus N(a)$ in $G_{x,w,y,a}$\;\label{3a-path}
      \If{if the path exists and is shorter than any observed
        so far}{
        remember $w,x,y$ and the path\;
      }
    }
  }
  \Return {subgraph of $G(M)$ induced by the remembered set of
    vertices (if any)}
\end{algorithm}

\emph{Correctness of Algorithm~\ref{alg:type3a}.}  Let us first verify
that the vertices of a shortest path found in line~\ref{3a-path} and
$w,x,y$ induce a subgraph of type III. Obviously, $x$ is connected
only to $w$, $w$ is not connected to $y$ and the last vertex $z$ of
the path. On the other hand, $w$ must be connected to all other white
vertices on the path, since any such white vertex that is not in
$N(w)$ is in $\C \setminus N(a)$ and hence, also $\C \setminus
N(w)\setminus N(a)$, i.e., we would have a shorter path ending at this
vertex. Since the path is a shortest path, all black vertices on the
path are connected only to its predecessor and successor on the
path. In addition $a$ is connected to $y$ and no other black vertex on
the path is connected to $y$ since $G_{x,w,y,a}$ does not contain any
other neighbors of $y$. It follows that the vertices $w,x,y$ and the
vertices of a shortest path induce a subgraph of type III.

Second, consider a smallest subgraph of type III in $G(M)$. We will
show it is considered by the algorithm. Assume the algorithm is in the
cycle, where it picked edges $\{x,w\}$ and $\{y,a\}$ of this
subgraph. Then the rest of the vertices must lie in $G_{x,w,y,a}$: the
remaining black vertices are not connected to $x$ and $y$ and the
remaining white vertices are either in $N(w)\setminus \{x\}$ and $z$
is $\C \setminus N(a)$. These vertices together with $a$ must form a
shortest path from $a$ to $\C \setminus N(w)\setminus N(a)$ in
$G_{x,w,y,a}$, hence, Algorithm~\ref{alg:type3a} finds this subgraph
or a subgraph with the same number of vertices.

\emph{Complexity of Algorithm~\ref{alg:type3a}.} We will show that the
complexity of Algorithm~\ref{alg:type3a} is $O(e^{3}) =
O(\Delta^{3}m^{3})$. The first loop executes $e$ times. The second
loop executes $O(e)$ times. Constructing graph $G_{x,w,y,a}$ takes
time $O(e)$. Finding a shortest path in $G_{x}$ can be done in time
 O(e)  using a breadth-first search algorithm.

\subsection{Type IV}
\label{sec:type-iv}

Algorithm~\ref{alg:type4} determines if $G(M)$ contains a forbidden
subgraph of type IV in time $O(m^{3}e)$.

\begin{algorithm}
  \caption{Find a $G_{\mathrm{IV}}$ subgraph.}
  \label{alg:type4}
  \scriptsize
  \SetKwInOut{Input}{Input}
  \SetKwInOut{Output}{Output}
    \Input {$G(M)$}
    \Output {A subgraph $G_{\mathrm{IV}}$ of $G(M)$}
    \BlankLine
    \For{distinct $a,b,c,d\in R$}{
      find $UX = N(a)\setminus (N(b)\cup N(c))$\;
      find $VY = N(b)\setminus (N(a)\cup N(c))$\;
      find $WZ = N(c)\setminus (N(a)\cup N(b))$\;
      find $U = UX\cap N(d)$ and $X = UX\setminus N(d)$\;
      find $V = VY\cap N(d)$ and $Y = VY\setminus N(d)$\;
      find $W = WZ\cap N(d)$ and $Z = WZ\setminus N(d)$\;
      \If{each of the sets $X,Y,Z,U,V,W$ is non-empty}{
        pick any $x\in X,y\in Y,z\in Z,u\in U,v\in V,w\in W$\;
        \Return{$G(M)[a,b,c,d,x,y,z,u,v,w]$}
      }
    }
    \Return {not found}
\end{algorithm}

\emph{Correctness of Algorithm~\ref{alg:type4}.} It is easy to see
that once $a,b,c,d$ are picked, each of $x,y,z,u,v,w$ has to belong to
computed set $X,Y,Z,U,V,W$, respectively, and that once they are
picked from those sets, the returned vertices induce
$G_{\mathrm{IV}}$.

\emph{Complexity of Algorithm~\ref{alg:type4}.} We will show that the
complexity of Algorithm~\ref{alg:type4} is $O(m^{3}e) = O(\Delta
m^{4})$. The time complexity of the steps inside the loop depends on
degrees of nodes $a,b,c,d$, i.e., it is $O(\deg (a) + \deg (b) + \deg (c)
+ \deg (d))$. Hence, the overall complexity is $\sum_{a,b,c,d\in R}
O(\deg (a) + \deg (b)\deg (c) + \deg (d)) = 4\sum_{a,b,c,d\in R}
O(\deg (d)) = 4\sum_{a,b,c\in R} O(e) = m^{3}e$.

\subsection{Type V}
\label{sec:type-v}

Algorithm~\ref{alg:type5} determines if $G(M)$ contains a forbidden
subgraph of type V in time $O(m^{3}e)$.

\begin{algorithm}
  \caption{Find a $G_{\mathrm{V}}$ subgraph.}
  \label{alg:type5}
  \scriptsize
  \SetKwInOut{Input}{Input}
  \SetKwInOut{Output}{Output}
    \Input {$G(M)$}
    \Output {A subgraph $G_{\mathrm{V}}$ of $G(M)$}
    \BlankLine
    \For{distinct $a,b,c,d\in \R$}{
      find $UY = N(b)\cap N(d)\setminus N(c)$\;
      find $VZ = N(b)\cap N(c)\setminus N(d)$\;
      find $U = UY\cap N(a)$ and $Y = UY\setminus N(a)$\;
      find $V = VZ\cap N(a)$ and $Z = VZ\setminus N(a)$\;
      find $X = N(a)\setminus (N(b)\cup N(c)\cup N(d))$\;
      \If{each of the sets $X,Y,Z,U,V$ is non-empty}{
        pick any $x\in X,y\in Y,z\in Z,u\in U,v\in V$\;
        \Return{$G(M)[a,b,c,d,x,y,z,u,v]$}
      }
    }
    \Return {not found}
\end{algorithm}

\emph{Correctness of Algorithm~\ref{alg:type5}.} It is easy to see
that once $a,b,c,d$ are picked, each of $x,y,z,u,v$ has to belong to
computed set $X,Y,Z,U,V$, respectively, and that once they are picked
from those sets, the returned vertices induce $G_{\mathrm{V}}$.

\emph{Complexity of Algorithm~\ref{alg:type5}.} The complexity of
Algorithm~\ref{alg:type5} is $O(m^{3}e) = O(\Delta m^{4})$. This
follows by the same argument as for Algorithm~\ref{alg:type5}.

\section{Detection of a smallest forbidden subgraph}
\label{sec:detect-small-forb-all}

Overall, we will use Dom et al. (\cite{DomGuoNiedermeier2010})
approach to find the smallest forbidden subgraph in $G(M)$. We will
first find a shortest-paths (the sum of the lengths of the three paths)
white asteroidal triple $A$ in time $O(n^{2}e) = O(\Delta mn^{2})$
using the algorithm in \cite{DomGuoNiedermeier2010}.



A shortest-paths white asteroidal triple $A$
must be in $T$, but does not need to be a smallest forbidden
subgraph. Let $\ell $ be the sum of the lengths of the three paths of
$A$. If $A$ is of
\begin{itemize}
\item type I or II, then it contains $\ell $ vertices;
\item type III, it contains $\ell - 5$ vertices;
\item type IV, it contains $10 = \ell - 8$ vertices;
\item type V, it contains $9 = \ell - 1$ vertices.
\end{itemize}
It follows that if one of the smallest forbidden subgraphs is of type
I or II, then each shortest-paths asteroidal triple is of type I or II and
is a smallest forbidden subgraph. For the remaining cases, we need
to determine the smallest forbidden subgraphs of type III, IV and
V. However, we only need to find a smallest subgraph of type X if it is a
smallest forbidden subgraph. Hence, for types IV and V, if we find
during the search that there is a smaller forbidden subgraph of some other
type, we can stop searching for this type. For type III, since it has
a variable size, we cannot stop searching, however, we can abandon the
branch which would yield a larger or even the same size subgraph of
type III than we have observed. We will use this in what follows to
obtain faster algorithms for types III, IV and V than the ones
presented in the previous section.

\subsection{Type III}
\label{sec:type-iii-all}

Algorithm~\ref{alg:type3-all} guarantees to find a smallest subgraph
of type III \textbf{if} it is smaller than other types of forbidden
subgraphs in time $O(ne^{2})$. If there is a smaller subgraph of type
I or there is a smaller of same size subgraph of type V in $G(M)$, it
either reports that or it could report a subgraph of type III which is
not the smallest. It will first determine whether
$G_{\mathrm{III}_{1}}$ is a subgraph of $G(M)$. If not it continues to
the second phase, where it assumes that the smallest subgraph of type
III (if it exists) has at least 9 vertices.

\begin{algorithm}[H]
  \caption{Find a smallest $G_{\mathrm{III}_{k}}$ subgraph if it is
    smaller than other types of subgraphs.}
  \label{alg:type3-all}
  \scriptsize
  \SetKwInOut{Input}{Input}
  \SetKwInOut{Output}{Output}
  \Input {$G(M)$}
  \Output {A smallest subgraph $G_{\mathrm{III}_{k}}$ of $G(M)$ or
    report there is a subgraph of other type (I or V) of equal
    or smaller size}
  \BlankLine
  \For{$w\in \R $}{
    \For{$x,u\in N(w)$}{\label{3-all-ph1-b}
      construct the subgraph $G_{x,w,u}$ of $G(M)$ induced by vertices
      $N(u)\setminus N(x)\cup \C \setminus N(w)$\;
      find induced matching of size two using
      Lemma~\ref{l:induced-matching-2}\;
      \If{induced matching exists}{
        \Return subgraph of $G(M)$ induced by $x,w,u$ and the induced
        matching ($G_{\mathrm{III}_{1}}$)
      }\label{3-all-ph1-e}
    }
  }
  \tcc{We can now assume that there is no $G_{\mathrm{III}_{1}}$ in
    $G(M)$}
  set $i_{min} = \infty $\;
  \For{$\{x,w\}\in E_{M}$, where $x\in \C$ and $w\in \R$}{
    find $D = N_{2}(w)\setminus N(x)$ and $Y = N(D)\setminus N(w)$\;
    \For{$y\in$ Y}{\label{3-all-ph2-y-b}
      construct the subgraph $G_{x,w,y}$ of $G(M)$ induced by vertices
      $N(w)\setminus \{x\}\cup \{y\} \cup D$\;
      find $D_{i} = N_{i}(y)$ in $G_{x,w,y}$, for $i\ge 1$\;
      find $Y' = \{y'\in Y:\ D_{1}\setminus N(y')\ne \emptyset \} $
      and $D' = D\cap N(Y')$\;
      find smallest odd $i\ge 3$ such that $D_{i}\cap D'\ne \emptyset $
      (if possible)\;
      \If{found}{
        pick any $d_{i}\in D_{i}\cap D'$, any $y'\in Y'\cap N(d)$ \;
        find a path $P$ from $d_{i}$ to some $d_{1}\in D_{1}$ in
        $G_{w,w,y}$ of length $i - 1$\;
        \If{$\{y',d_{1}\} \notin E(M)$ and $i < i_{min}$}{
          set $i_{min}$ to $i$\;
          remember $x,w,y,y'$ and vertices of $P$\;
        }
      }
    }\label{3-all-ph2-y-e}
  }
  \eIf{$i_{min} = \infty $}{
    \Return subgraph of type III not found or there is a subgraph of
    type I or V of the size at most the size of the smallest type III subgraph
  }{
    \Return {subgraph of $G(M)$ induced by remembered set of vertices}
  }
\end{algorithm}

\emph{Correctness of Algorithm~\ref{alg:type3-all}.} It is easy to
check that the first phase of the algorithm finds
$G_{\mathrm{III}_{1}}$ subgraph if it exists in $G(M)$. Assume that
$G_{\mathrm{III}_{1}}$ is not an induced subgraph of $G(M)$., i.e.,
that a smallest subgraph of type III (if it exists) has at least 9
vertices. The algorithm continues to the second phase.

First, assume that $i$ is not found, i.e., for all odd $i\ge 3$,
$D_{i}\cap D' = \emptyset $. This implies that any path starting at
$y$ in $G_{x,w,y}$ cannot be extended with a white vertex $y'$ that is
not adjacent to $w$ and not adjacent to the second vertex $d_{1}\in
D_{1}$ of this path. Hence, the algorithm correctly continues with
examining another selection of vertices $x,w,y$. Assume that $i$ was
found. Now, assume that $G(M)$ does not contain edge
$\{y',d_{1}\}$. Let us verify that vertices $x,w,y,y'$ and the
vertices of $P$ induce $G_{\mathrm{III}_{(i - 1)/2}}$. It is clear
that $x$ is connected only to $w$ and $w$ only to white vertices on
$P$ except the first vertex $y$. By the construction, each vertex on
$P$ can be adjacent only to its predecessor or successor on $P$. Since
$i$ is the smallest odd integer larger than two such that $D_{i}\cap
D'\ne \emptyset $, $y'$ is not adjacent to any black vertex on the
path other than the last one. Hence, the vertices induce a subgraph of
type III. Finally, assume that $\{y',d_{1}\}\in E(M)$. If $i\ge 5$,
then vertices of $P$ without $y$ and $y'$ induce a cycle of length $i
+ 1$, i.e., a subgraph $G_{\mathrm{I}_{(i - 3)/2}}$, which is smaller
than a subgraph of type III we could get for this selection of $x,w,y$
(by choosing a different $d_{i}$, $y'$ or path $P$, or searching for
another odd $i$ such that $D_{i}\cap D'\ne \emptyset $). If $i = 3$,
consider $d_{1}'\in D_{1}$ that is not adjacent to $y'$ and let $P =
y,d_{1},u,d_{3}$. If $d_{1}'$ is adjacent to $u$, vertices
$x,w,u,d_{1}',y,d_{3},y'$ induce $G_{\mathrm{III}_{1}}$, a
contradiction. Hence, assume $\{d_{1}',u\} \notin E(M)$. Since
$d_{1}'\in D\subseteq N_{2}(w)$, there exists $u'\in N(w)$ adjacent to
$d_{1}'$. If $\{d_{1},u'\} \in E(M)$, then vertices
$x,w,u,u',d_{1},d_{1}',d_{3},y,y'$ induce $G_{\mathrm{V}}$. Otherwise,
vertices $w,u,d_{1},y,d_{1}',u'$ induce a cycle of length 6. In any
case, there exists a subgraph of other type of size equal or smaller
than it would be possible to find for this choice of $x,w,y$, hence,
the algorithm correctly moves to the next choice.

\emph{Complexity of Algorithm~\ref{alg:type3-all}.} We will show that
the complexity of Algorithm~\ref{alg:type3-all} is $O(ne^{2}) =
O(\Delta^{2}m^{2})$. The body of the loop in
lines~\ref{3-all-ph1-b}--\ref{3-all-ph1-e} will execute $O(\Delta e)$
times and each step of the body take $O(e)$ time. Hence, the
complexity of the first phase is $O(\Delta e^{2}) = O(ne^{2})$.
The main loop of the second phase will execute $O(e)$
times. Determining $D$ and $Y$ takes time $O(e)$. The nested loop in
lines~\ref{3-all-ph2-y-b}--\ref{3-all-ph2-y-e} will execute $O(n)$
times. Each step of the body of this loop will take time
$O(e)$. Hence, the complexity of the second phase is $O(ne^{2})$.

\subsection{Type IV}
\label{sec:type-iv-all}

Algorithm~\ref{alg:type4a-all} finds the subgraph $G_{\mathrm{IV}}$ in
time $O(n^{3}e)$, if it exists and if it is a smallest forbidden
subgraph. If there is a smaller forbidden subgraph of type I or III,
it might find an instance of $G_{\mathrm{IV}}$ or it might report
that there is a smaller forbidden subgraph instead.



\begin{algorithm}
  \caption{Find a $G_{\mathrm{IV}}$ subgraph or report that there is a
    smaller subgraph of type I or III.}
  \label{alg:type4a-all}
  \scriptsize
  \SetKwInOut{Input}{Input}
  \SetKwInOut{Output}{Output}
    \Input {$G(M)$}
    \Output {A subgraph $G_{\mathrm{IV}}$ of $G(M)$ or report that
      $G_{\mathrm{IV}}$ is not a smallest subgraph}
    \BlankLine
    \For{distinct $x,y,z\in \C$}{
      find $A = N(x)\setminus (N(y)\cup N(z))$\;
      find $B = N(y)\setminus (N(x)\cup N(z))$\;
      find $C = N(z)\setminus (N(x)\cup N(y))$\;
      find $D = \C\setminus (N(x)\cup N(y)\cup N(z))$\;
      find $U = N(A)\setminus \{x,y,z\}$\;
      find $V = N(B)\setminus \{x,y,z\}$\;
      find $W = N(C)\setminus \{x,y,z\}$\;
      \If{all sets $A,B,C,D,U,V,W$ are non-empty}{
        \For{$d\in D$}{
          \If{there exists distinct $u\in U\cap N(d)$, $v\in V\cap
            N(d)$ and $w\in W\cap N(d)$}{
            find $a\in A\cap N(u)$, $b\in B\cap N(v)$ and $c\in C\cap
            N(w)$\;
            \eIf{none of the edges
              $\{a,v\},\{a,w\},\{b,u\},\{b,w\},\{c,u\},\{c,v\}$
              exists}{
              \Return {$G(M)[x,y,z,u,v,w,a,b,c,d] = G_{\mathrm{IV}}$}
            }{
              \Return {there is a smaller subgraph of type I or III}
            }
          }
        }
      }
    }
    \Return {not found}
\end{algorithm}

\emph{Correctness of Algorithm~\ref{alg:type4a-all}.}
Correctness of the algorithm follows by the following lemma.

\begin{lemma}
  \label{l:cross-edges}
  Consider a subgraph $G'$ of $G(M)$ induced by vertices
  $x,y,z,u,v,w,a,b,c,d$ that contains edges
  \begin{equation*}
    \{x,a\} ,\{y,b\} ,\{z,c\} ,\{a,u\} ,\{b,v\} ,\{c,w\} ,\{u,d\} ,\{v,d\} ,\{w,d\} \,,
  \end{equation*}
  and does not contain edges
  \begin{equation*}
    \{x,d\} ,\{y,d\} ,\{z,d\} \,.
  \end{equation*}
  Then either $G'$ is an instance of $G_{\mathrm{IV}}$ or $G'$ contains
  either $G_{\mathrm{I}_{1}}$, $G_{\mathrm{III}_{1}}$ or
  $G_{\mathrm{III}_{2}}$ as an induced subgraph.
\end{lemma}

\begin{proof}
  We will use the following two partial maps: $R(x) = a$, $R(y) = b$,
  $R(z) = c$, $R(a) = u$, $R(b) = v$ and $R(c) = w$, and $L = R^{-1}$.

  If none of the edges in
  $E' = \{\{a,v\},\{a,w\},\{b,u\},\{b,w\},\{c,u\},\{c,v\}\}$ is present, then
  $G'$ is isomorphic to $G_{\mathrm{IV}}$.

  If exactly one edge $e$ in $E'$ is present, we have an induced
  subgraph $G_{\mathrm{III}_{1}}$ centered at the vertex $r = e\cap
  \{u,v,w\}$. In particular, vertices $d,r,L(r),L(L(r)),\ell ,L(\ell
  ),z$, where $\ell = e \cap \{a,b,c\}$ and $z\in \{u,v,w\} \setminus
  \{r,R(\ell )\}$, induce $G_{\mathrm{III}_{1}}$.

  We can assume that there are at least two edges in $E'$ present. We
  will distinguish two cases. Either (i) there exists two edges $e$ and
  $e'$ in $E'$ present such that $e\cap e'\ne \emptyset $, or (ii) for
  each pair of such edges $e\cap e' = \emptyset $.

  First, consider case (i) and let $e,e'$ be such that $e\cap e'\ne \emptyset
  $. Depending on whether the intersection lies in $\{a,b,c\}$ or
  $\{u,v,w\}$, we have two cases:
  \begin{enumerate}
  \item $e\cap e'\in \{a,b,c\}$ (``edges joing on the left''), then
    vertices $V(G')\setminus \{e\cap e'\}$ induce
    $G_{\mathrm{III}_{2}}$;
  \item $e\cap e'\in \{u,v,w\}$ (``edges joing on the right''), then
    vertices $x,y,z,a,b,c,e\cap e'$ induce $G_{\mathrm{III}_{1}}$.
  \end{enumerate}

  Now, consider case (ii). Note the number of edges in $E'$ present is
  at most three. We will consider two cases depending on the number of
  such edges:
  \begin{enumerate}
  \item $|E'\cap E(G')| = 2$: Without loss of generality we can assume
    that $e\cap \{a,b,c\} = L(e'\cap \{u,v,w\})$ for $e,e'\in E'$
    present in $G'$. Then the same collection of vertices as in the
    case of one edge $e$ induces $G_{\mathrm{III}_{1}}$, since one end
    of $e'$ lies outside of this collection.
  \item $|E'\cap E(G')| = 3$: Then the vertices $a,b,c,u,v,w$ induce
    $C_{6}$, i.e., $G_{\mathrm{I}_{1}}$.
  \end{enumerate}
\end{proof}

\emph{Complexity of Algorithm~\ref{alg:type4a-all}.}
We will show that the complexity of Algorithm~\ref{alg:type4a-all} is
$O(n^{3}e) = O(\Delta mn^{3})$. The first loop executes
$O(n^{3})$ times, determining $A,B,C,D$ takes time $O(m)$, determining
sets $U,V,W$ time $O(e)$. The loop for $d\in D$ is executed $O(m)$
times and each execution takes time $O(deg(d))$, i.e., the total time
spent in this loop is $\sum_{d\in D} O(\deg (d)) = O(e)$.

\subsection{Type V}
\label{sec:type-v-all}

Algorithm~\ref{alg:type5a-all} find the subgraph $G_{\mathrm{V}}$ in
time $O(n^{3}e)$, if it exists and if it is a smallest forbidden
subgraph. If there is a smaller forbidden subgraph of type I or
III, it might find an instance of $G_{\mathrm{V}}$ or it might
report that there is a smaller forbidden subgraph instead.



\begin{algorithm}
  \caption{Find a $G_{\mathrm{V}}$ subgraph or report that there is a
    smaller subgraph of type I or III.}
  \label{alg:type5a-all}
  \scriptsize
  \SetKwInOut{Input}{Input}
  \SetKwInOut{Output}{Output}
    \Input {$G(M)$}
    \Output {A subgraph $G_{\mathrm{V}}$ of $G(M)$ or report that
      $G_{\mathrm{V}}$ is not a smallest subgraph}
    \BlankLine
    \For{distinct $x,y,z\in \C$}{
      find $A = N(x)\setminus (N(y)\cup N(z))$\;
      find $B = N(y)\setminus (N(x)\cup N(z))$\;
      find $C = N(z)\setminus (N(x)\cup N(y))$\;
      find $D = (N(y)\cap N(z))\setminus N(x)$\;
      pick any $u\in N(A)\cap N(B)\cap N(D)$ if possible\;
      pick any $v\in N(A)\cap N(C)\cap N(D)$ if possible\;
      \If{$u$ and $v$ has been picked}{
        \If{$u\in N(C)$ or $v\in N(B)$}{
          \Return there is a smaller subgraph of type III  ($G_{\mathrm{III}_{1}}$)\label{5a-all-III}
        }
        find $A' = A\cap N(u)\cap N(v)$ and $D' = D\cap N(u)\cap N(v)$\;
        \If{$A' = \emptyset $ or $D' = \emptyset $}{
          \Return there is a smaller subgraph of type I
          ($G_{\mathrm{I}_{1}}$ or $G_{\mathrm{I}_{2}}$)\label{5a-all-I}
        }
        pick any $a\in A'$, $b\in B\cap N(u)$, $c\in C\cap N(v)$ and
        $d\in D'$\;
        \Return $G(M)[x,y,z,u,v,a,b,c,d]$\label{5a-all-V}
      }
    }
    \Return {not found}
\end{algorithm}

\emph{Correctness of Algorithm~\ref{alg:type5a-all}.}
The algorithm is able
to reduce time complexity by avoiding trying all possible choices for
$u,v$ and $a,b,c,d$, but rather picking one choice (if
possible), and then either finding $G_{\mathrm{V}}$ or a smaller
forbidden subgraph. Let us verify that decisions algorithm makes are
correct:
\begin{itemize}
\item
  First, assume that the algorithm stops in
  line~\ref{5a-all-III}. Then there exists $w\in N(A)\cap N(B)\cap
  N(C)\cap N(D)$ (either $u$ or $v$). Then there exists $a\in A\cap
  N(w)$, $b\in B\cap N(w)$ and $c\in C\cap N(w)$. Vertices
  $x,y,z,a,b,c,w$ induce $G_{\mathrm{III}_{1}}$.
\item Assume that the algorithm stops in line~\ref{5a-all-I}. If $A' =
  \emptyset $ and $D' = \emptyset $, there exists $a\in A\cap N(u)$,
  $a'\in A\cap N(v)$, $d\in D\cap N(u)$ and $d'\in D\cap N(v)$. Note
  that $a\ne a'$, $d\ne d'$, $a,d\notin N(v)$ and $a',d'\not in$
  N(u). It is easy to check that vertices $x,a,u,d,y,d',v,a'$ induce
  $C_{8}$. Similarly, if either $A' = \emptyset $ or $D' = \emptyset
  $, we can find vertices that induce $C_{6}$.
\item
  Finally, it is easy to check that if the algorithm outputs an
  induced subgraph in line~\ref{5a-all-V}, it is $G_{\mathrm{V}}$.
\end{itemize}
On the other hand, if $G_{\mathrm{V}}$ is a smallest forbidden
subgraph of $G(M)$, then the algorithm cannot finish in
lines~\ref{5a-all-III} and~\ref{5a-all-I}, and
hence, it will eventually output $G_{\mathrm{IV}}$ in
line~\ref{5a-all-V}.

\emph{Complexity of Algorithm~\ref{alg:type5a-all}.}  We will show
that the complexity of Algorithm~\ref{alg:type5a-all} is $O(n^{3}e) =
O(\Delta mn^{3})$. The first loop executes $O(n^{3})$ times,
determining $A,B,C,D$ takes time $O(m)$, picking $u,v$ time
$O(e)$, picking $a,b,c,d$ time $O(m)$. Hence, the total time
used by the algorithm is $O(n^{3}(O(m) + O(e))) = O(n^{3}e)$.

\subsection{Main algorithm}
\label{sec:main-algorithm}

Algorithm~\ref{alg:Tucker} finds a smallest forbidden subgraph
using the three algorithms described above.

\begin{algorithm}
  \caption{Find a smallest forbidden Tucker subgraph.}
  \label{alg:Tucker}
  \scriptsize
  \SetKwInOut{Input}{Input}
  \SetKwInOut{Output}{Output}
    \Input {$G(M)$}
    \Output {A smallest forbidden subgraph of $G(M)$}
    \BlankLine
    find a smallest white asteroidal triple $A$ using
    Lemma~\ref{l:asteroidal}\;
    let $\ell $ be the sum of the lengths of three paths of $A$\;
    find a smallest subgraph of types III, IV and V (using the
    procedures described above)\;
    let $s_{\mathrm{III}},s_{\mathrm{IV}},s_{\mathrm{V}}$ be
    the sizes of these subgraphs (or $\infty $ if not found), respectively\;
    \eIf {$\ell = \min\{\ell ,s_{\mathrm{III}},s_{\mathrm{IV}},s_{\mathrm{V}}\}$}
    {\Return {$A$}}
    {let $s_{X} =
      \min\{\ell ,s_{\mathrm{III}},s_{\mathrm{IV}},s_{\mathrm{V}}\}$\;
    \Return {the smallest subgraph of type $X$}}
\end{algorithm}

To verify the correctness of Algorithm~\ref{alg:Tucker}, first
consider that one of the smallest forbidden subgraphs of $G(M)$ is of
type I or II. By the above argument, asteroidal triple $A$ is of type
I or II with size $\ell $, and since it is a smallest forbidden
subgraph, we have $\ell =
\min\{\ell,s_{\mathrm{III}},s_{\mathrm{IV}},s_{\mathrm{V}}\}$. Hence,
the algorithm correctly outputs one of the smallest forbidden
subgraphs. Second, assume that all smallest forbidden subgraphs of
$G(M)$ are of type III, IV and V. Let $s =
\min\{s_{\mathrm{III}},s_{\mathrm{IV}},s_{\mathrm{V}}\}$. If $A$ is of
type I or II, then the size of $A$ is $\ell $, and hence, $\ell > s$
and $s_{X} =
\min\{\ell,s_{\mathrm{III}},s_{\mathrm{IV}},s_{\mathrm{V}}\}$. If $A$
is of type III, IV or V, then $\ell \ge s + 1$, and hence again $s_{X}
= \min\{\ell,s_{\mathrm{III}},s_{\mathrm{IV}},s_{\mathrm{V}}\}$. It
follows that Algorithm~\ref{alg:Tucker} correctly outputs one of the
smallest forbidden subgraphs.

It follows from Algorithm~\ref{alg:Tucker} that we do not need a
special detection algorithms for type I and II forbidden
subgraphs. However, in some applications, there might be a need to
determine a smallest forbidden subgraph of each type. Therefore, we
present such algorithms for these two types of forbidden subgraphs as
well.


\begin{thebibliography}{99}

\bibitem{adam-modelfree} Zaky Adam, Monique Turmel, Claude Lemieux, David Sankoff: {\em Common Intervals and Symmetric Difference in a
Model-Free Phylogenomics, with an Application to Streptophyte
Evolution}. Journal of Computational Biology 14(4): 436-445 (2007)


\bibitem{alizadeh-physical} Farid Alizadeh, Richard M. Karp, Lee Aaron Newberg, Deborah K.
Weisser: {\em Physical Mapping of Chromosomes: A Combinatorial
Problem in Molecular Biology}. Algorithmica 13(1/2): 52-76 (1995)

\bibitem{BlinRizziVialette2012} Guillaume Blin and Romeo Rizzi and St\'{e}phane
Vialette: {\em A Faster Algorithm for Finding Minimum {Tucker}
Submatrices}. J. Theory Comput. Syst. (51) 270-281 (2012)


\bibitem{BoothLueker1976} Kellogg S. Booth and George S. Lueker: {\em Testing for the Consecutive Ones Property,
Interval Graphs, and Graph Planarity Using PQ-Tree Algorithms}. J.
Comput. Syst. Sci. 13(3): 335-379 (1976)


\bibitem{chauve-methodological} Cedric Chauve, Eric Tannier: {\em A Methodological Framework for the Reconstruction of Contiguous Regions of Ancestral Genomes and Its Application to
Mammalian Genomes}. PLoS Computational Biology 4(11) (2008)

\bibitem{DomGuoNiedermeier2010} Michael Dom and Jiong Guo and Rolf
Niedermeier : {\em Approximation and fixed-parameter algorithms
for consecutive ones submatrix problems} J.Computer and System
Sciences 76 (3-4) 204-221 (2010)



\bibitem{habib-lex} M.Habib, Ross M. McConnell, Christophe Paul, Laurent Viennot: {\em Lex-BFS and partition refinement,
with applications to transitive orientation, interval graph
recognition and consecutive ones testing}. Theor. Comput. Sci.
234(1-2): 59-84 (2000)

\bibitem{hsu-simple} Wen-Lian Hsu: {\em A Simple Test for the Consecutive Ones Property.} J. Algorithms 43(1): 1-16 (2002)

\bibitem{lu-test} Wei-Fu Lu, Wen-Lian Hsu: {\em A Test for the Consecutive Ones Property on Noisy Data -
Application to Physical Mapping and Sequence Assembly}. Journal of
Computational Biology 10(5): 709-735 (2003)

\bibitem{DBLP:conf/wg/LindzeyM13} Nathan Lindzey and Ross M. McConnell : {\em On Finding Tucker Submatrices and Lekkerkerker-Boland
               Subgraphs}. WG 2013.

\bibitem{ma-reconstructing} Jian Ma, Louxin Zhang, Bernard B. Suh, Brian J. Raney, Richard C. Burhans, W. James Kent, Mathieu Blanchette, David Haussler, and Webb Miller1 :
{\em Reconstructing contiguous regions of an ancestral genome}.
GenomeRes 16(12) 1557--1565 (2006)

\bibitem{mcconnell-certifying}
Ross M. McConnell: {\em A certifying algorithm for the
consecutive-ones property.} SODA 2004: 768-777

\bibitem{meidanis-on} Joao Meidanis, Oscar Porto, Guilherme P. Telles: {\em On the Consecutive Ones Property.} Discrete Applied Mathematics 88(1-3): 325-354 (1998)

\bibitem{Tucker1972} A. C. Tucker: {\em A structure theorem for the consecutive 1's property}.
J. of Comb. Theory, Series B 12 :153-162 (1972)








\end{thebibliography}

\end{document}